\documentclass{article}

\usepackage{arxiv}

\usepackage[utf8]{inputenc} 
\usepackage[T1]{fontenc}    
\usepackage{hyperref}       
\usepackage{url}            
\usepackage{booktabs}       
\usepackage{amsfonts}       
\usepackage{nicefrac}       
\usepackage{microtype}      
\usepackage{lipsum}		
\usepackage{graphicx}
\usepackage{mathtools}
\usepackage{amssymb}
\usepackage{amsthm}
\usepackage{amsmath}
\usepackage{doi}

\usepackage{cases}
\usepackage{microtype}
\usepackage[utf8]{inputenx}
\usepackage{amssymb}
\usepackage{amsthm}
\usepackage{mathtools, cases}
\usepackage{mathbbol}
\usepackage{makecell}

\allowdisplaybreaks

\newtheorem{thm}{Theorem}

\title{Non-monotone dependence modeling with copulas: an application to the volume-return relationship}

\author{ 
	\href{https://orcid.org/0000-0002-6163-044X}{Manfred Marvin Marchione}\\
	Department of Statistical Sciences\\
	Sapienza University of Rome\\
	\texttt{manfredmarvin.marchione@uniroma1.it} \\
	\And
	\href{https://orcid.org/0000-0002-9926-7869}{Fabio Baione} \\
	Department of Statistical Sciences\\
	Sapienza University of Rome\\
	\texttt{fabio.baione@uniroma1.it}}

\date{\today}

\renewcommand{\headeright}{}
\renewcommand{\undertitle}{}

\begin{document}
\maketitle

\begin{abstract}
This paper introduces an innovative method for constructing copula models capable of describing arbitrary non-monotone dependence structures. The proposed method enables the creation of such copulas in parametric form, thus allowing the resulting models to adapt to diverse and intricate real-world data patterns. We apply this novel methodology to analyze the relationship between returns and trading volumes in financial markets, a domain where the existence of non-monotone dependencies is well-documented in the existing literature. Our approach exhibits superior adaptability compared to other models which have previously been proposed in the literature, enabling a deeper understanding of the dependence structure among the considered variables.
\end{abstract}

\keywords{Copulas \and Non-monotone dependence \and Measure-preserving transformations \and Volume-return relationship}

\section{Introduction}

\noindent While the introduction of copulas in the field of statistics dates back to the work by Sklar \cite{sklar}, according to Genest et al. \cite{genest} the scientific production on the topic has seen considerable growth only since 1999. A strong impetus for this growth was particularly given by the books by Joe \cite{joe} and Nelsen \cite{nelsenold}, which have greatly contributed to the spread of copula theory in the scientific community. The potential of copulas in financial and actuarial applications was highlighted by Embrechts \cite{embrechts} and Frees and Valdez \cite{frees}, and today copulas are a commonly used tool in modeling dependencies between random variables.
\noindent Most of the copulas known in the literature, however, are suited to describing phenomena that exhibit a monotonic type of dependence. Surprisingly, the use of copulas for describing non-monotonic dependencies has been rarely treated, and the literature on the subject appears limited. A mathematical analysis of the problem was proposed for the first time by Scarsini and Venetoulias \cite{scarsini}. The authors observed that arbitrary dependence structures can be obtained by applying transformations that preserve the uniform distribution to the marginals of a vector with uniform components. These transformations are known as \textit{measure-preserving} functions. In particular, if two uniformly distributed random variables exhibit maximum positive dependence, the application of a measure-preserving function leads to a bivariate distribution whose support coincides with the graph of the function itself. In the same work, the authors provided various examples of piecewise-linear functions which preserve the uniform distribution. In this paper, we introduce a method for constructing more general measure-preserving functions, thus enabling the construction of copulas which describe arbitrary non-monotone dependence structures. The proposed copula models are presented in a parametric form, which makes them highly adaptable for modeling the complicated dependence structures which emerge from real-world data.\\
The methodology of parameter estimation and model selection is illustrated by applying our results to the analysis of the relationship between returns and trading volumes in financial markets. The existence of a non-monotone dependence between returns and volumes is well-established in the literature. The birth of the scientific literature on the topic is generally attributed to the work of Osborne \cite{osborne}, who modeled the stock price as a diffusion process having variance dependent on the number of transactions. However, as pointed out by Karpoff \cite{karpoff}, Osborne assumed that the transactions are uniformly distributed in time and therefore he expressed the price process in terms of time intervals. Thus, he did not directly examine the relation between volume and returns. It was instead Clark \cite{clark} who modeled the randomness of the velocity of transactions by time-changing the price process with a Lévy stable subordinator. By treating the  trading volume as a measure of the number of transactions, the author concluded that a positive correlation between trading volume and the absolute value of the price variation is to be expected. This conclusion was empirically supported by analyzing data concerning prices and trading volumes of cotton futures. Over the years, a variety of empirical studies have been proposed by different authors. For instance, one such study by Crouch \cite{crouch} showed that the absolute value of price variation was correlated with the trading volume for stocks in the aerospace industry in the United States. Similar concusions were achieved by Tauchen and Pitts \cite{tauchen} in the study of the prices of the T-bills. An interesting category of research in this field comprises studies that have explored the causal relationships between returns and trading volumes, aiming to investigate potential practical applications. Lee and Rui \cite{lee} examined the existence of Granger causality relations between volumes and returns in the New York, Tokyo and London stock markets, while causal relationships in Eastern European markets were studied by G\"und\"uz and Hatemi-J \cite{gunduz}. Unfortunately, the authors have identified causal links only in rare instances, which often involved bidirectional causality. Therefore, although the literature on theoretical models and empirical studies is extensive and diverse, the identification of practical applications remains an area yet to be fully explored.\\
\noindent Credit for studying the non-monotone relationship between volumes and return using copulas for the first time must be attributed to Neto \cite{neto}. The author used the findings of Scarsini and Venetoulias \cite{scarsini} to study the relationship between returns and volume for stocks issued by some French companies. However, the copula model used by Neto is based on the assumption of a fixed dependence structure, which is not defined in a parametric form and hence it is not designed to adapt to the specific data under consideration. The results presented here allow for the construction of flexible copula models that can adapt to any dependence structure found in the data.

\section{Copulas and measure preserving functions}
\noindent This section is devoted to discussing the role of measure-preserving functions in the construction of copulas. We start our work by briefly recalling the definition and some basic properties of copulas. Throughout the paper, we will restrict our investigation to the bivariate case. Therefore, we say that a function $\text{C}:[0,1]^2\to[0,1]$ is a copula if it satisfies the following conditions:
\begin{enumerate}
\item $\text{C}(u,0)=\text{C}(0,v)=0$
\item $\text{C}(u,1)=u$ and $\text{C}(1,v)=v$
\item $\forall\;u_1, u_2, v_1, v_2\in[0,1]$ such that $u_1\le u_2$ and $v_1\le v_2$, it holds that $$\normalfont{\text{V}}_{\normalfont{\text{C}}}([u_1,u_2]\times[v_1,v_2]):=\text{C}(u_2,v_2)-\text{C}(u_2,v_1)-\text{C}(u_1,v_2)+\text{C}(u_1,v_1)\ge0.$$
\end{enumerate}
\noindent From a probabilistic point of view, copulas are cumulative distributions functions of multivariate distributions having standard uniform marginals. It is well-known that every copula $\text{C}$ satisfies, for all $u,v\in[0,1]$, the inequality $$\min(u,v)\le\text{C}(u,v)\le\max(u+v-1,0).$$
\noindent The functions $\text{M}(u,v)=\min(u,v)$ and $\text{W}(u,v)=\max(u+v-1,0)$ are called respectively Fréchet-Hoeffding upper and lower bound. It can be easily shown that the Fréchet-Hoeffding bounds are both copulas in the bivariate case. What is relevant from our point of view is the probabilistic interpretation of the Fréchet-Hoeffding bounds, which is established by the following standard result (see, for example, Fréchet \cite{frechet}).

\begin{thm}\label{montrans}Let $X_1$ and $X_2$ be two random variables such that $X_2=\phi(X_1)$ where $\phi$ is a stricly monotonic function. Denote by $F_{X_1}$ and $F_{X_2}$ the cumulative distribution functions of $X_1$ and $X_2$ respectively.\\
\noindent If $\phi$ is monotonically increasing then the cumulative distribution function of $(X_1,X_2)$ is $$F(x_1,x_2)=\min\left(F_{X_1}(x_1),\;F_{X_2}(x_2)\right).$$
If $\phi$ monotonically decreasing then
$$F(x_1,x_2)=\max\left(F_{X_1}(x_1)+\;F_{X_2}(x_2)-1,\;0\right).$$\end{thm}
\noindent Theorem \ref{montrans} essentially states that the Fréchet-Hoeffding upper and lower bounds describe perfect positive and negative dependence respectively, regardless of the functional form of the dependence and the marginal distributions. It can be easily verified that both the Fréchet-Hoeffding bounds $M$ and $W$ concentrate the probability mass of a bivariate uniform distribution on one of the diagonals of the unit square. This implies that theorem \ref{montrans} can be stated in an equivalent form by asserting that the only possible monotone dependence among two uniformly distributed random variables is the linear dependence. In particular, consider two random variables $U$ and $V$ uniformy distributed over the interval $[0,1]$. Assume that there exists a monotonically increasing function $g:[0,1]\rightarrow[0,1]$ such that $V=g(U)$. By studying the distribution of $V$, it is immediately verified that
\[\mathbb{P}\left(V\le v\right)=\mathbb{P}\left(U\le g^{-1}(v)\right)=g^{-1}(v),\qquad v\in[0,1].\]
\noindent Hence, it is clear that, unless $g$ coincides with the identity function, the assumption that $V$ has a standard uniform distribution is contradicted. Similar conclusions can be drawn in the case of perfect negative dependence. It follows that the only functions which can describe perfect monotone dependence among uniformly distributed random variables are linear. In particular, if \(U\) and \(V\) are two random variables with uniform distribution over \([0,1]\) and \(g: [0,1] \rightarrow [0,1]\) is a strictly monotonic function, the relation \(V=g(U)\) can only hold if \(g(u)=u\) or \(g(u)=1-u\), in which cases the points with coordinates \((U,V)\) will be concentrated on the curve coinciding with the graph of \(g\).\\

\noindent The purpose of this paper is to generalize the ideas described so far to the case where \(g\) is not a monotonic function. From this point forward, we will indicate that a random variable $U$ is uniformly distributed over the interval $[0,1]$ with the notation $U\sim\mathcal{U}_{[0,1]}$. We start our analysis by introducing the notion of \textit{measure-preserving function}.  We say that a function \(g: [0,1] \rightarrow [0,1]\) is measure-preserving if, for a random variable \(U\sim\mathcal{U}_{[0,1]}\), it holds that $g(U)\sim\mathcal{U}_{[0,1]}$. Equivalently, we say that the function $g$ preserves the uniform distribution.
\vspace{2mm}

\noindent We now observe that, if a function \(g:[0,1]\rightarrow[0,1]\) is measure-preserving and $U\sim\mathcal{U}_{[0,1]}$, the pair \((U,g(U))\) follows a distribution with uniform marginals and therefore its cumulative distribution function is a copula. Moreover, the dependence is perfect and it is immediate to verify that the probability mass is concentrated on the curve that coincides with the graph of \(g\). Of course, this is true regardless of whether \(g\) is a monotonic function. Consequently, any function that preserves the uniform distribution corresponds to a curve on which the probability mass of an appropriate copula describing perfect dependence is concentrated. These observations are at the heart of the work by Scarsini and Venetoulias \cite{scarsini}. Following the ideas of the same authors, it is therefore possible to construct copulas that describe dependencies of arbitrary type in the following way.
\vspace{2mm}

\noindent Consider a pair \((U,V)\) of random variables, each uniformly distributed over the interval \([0,1]\) and joined by the copula \(\normalfont\textrm{C}\). Assume that \(g\) is a measure-preserving function. Thus, the pair \((U, g(V))\) also exhibits standard uniform marginal distributions. This occurs because applying the function \(g\) to one of the components of \((U,\;V)\) does not alter the marginal distributions and ensures that the resulting distribution is still a copula. However, transforming one of the two components of the initial vector alters the dependence structure, modifying the copula \(\normalfont\textrm{C}\) into a new copula which we denote by \(\normalfont\textrm{Q}\).  In principle, by choosing an initial copula \(\normalfont\textrm{C}\) and applying an appropriate function \(g\) that preserves the uniform distribution, it is theoretically possible to construct distributions with arbitrary dependence structures. Assume for simplicity that the copula joining the pair \((U, V)\), with $U,V\sim\mathcal{U}_{[0,1]}$, depends on a vector of parameters \(\boldsymbol{\theta}\in\mathbb{R}^d\) for some positive integer $d$, and denote this copula by \(\normalfont\textrm{C}_{\boldsymbol{\theta}}\). Moreover, assume that \(\normalfont\textrm{C}_{\boldsymbol{\theta}}\) encompasses the Fréchet-Hoeffding upper bound, either as a special or a limiting case. Specifically, we suppose that for some \(\boldsymbol{\theta}_0\in\mathbb{R}^d\) it holds that \[\lim_{\boldsymbol{\theta}\to\boldsymbol{\theta}_0}\normalfont\textrm{C}_{\boldsymbol{\theta}}(u,v)=M(u,v)\qquad\forall u,v\in[0,1].\] This means that the closer \(\boldsymbol{\theta}\) is to \(\boldsymbol{\theta}_0\), the more the points with coordinates \((U,V)\) tend to concentrate near the diagonal of the unit square. If we now denote by \(\normalfont\textrm{Q}_{\boldsymbol{\theta}}\) the copula that links \(U\) and \(g(V)\), it is clear that, as \(\boldsymbol{\theta}\) tends to \(\boldsymbol{\theta}_0\), the points with coordinates \((U,g(V))\) tend to concentrate near the graph of \(g\). In other words, the graph of \(g\) represents for \(\normalfont\textrm{Q}_{\boldsymbol{\theta}}\) the curve along which the points tend to concentrate when the dependence becomes stronger. Scarsini and Venetoulias \cite{scarsini} proposed the following measure preserving function as an example:
\begin{equation}g(v)=\begin{dcases}
\dfrac{(1-\zeta+\beta)}{\beta}\;v &\textrm{se }0\le v\le\beta
\\(1-\zeta+\beta)+ \dfrac{(\zeta-\beta)}{(\gamma-\beta)}(v-\beta) \qquad&\textrm{if }\beta< v\le\gamma
\\(1-\zeta+\beta)+ \dfrac{(\zeta-\beta)}{(\delta-\gamma)}(\delta-v) &\textrm{if }\gamma< v\le\delta
\\(1-\zeta+\beta)+ \dfrac{(\zeta-\beta)}{(\epsilon-\delta)}(v-\delta) &\textrm{if }\delta< v\le\epsilon
\\(1-\zeta+\beta)+ \dfrac{(\zeta-\beta)}{(\zeta-\epsilon)}(\zeta-v) &\textrm{if }\epsilon< v\le\zeta
\\\dfrac{(1-\zeta+\beta)}{(1-\zeta)}(1-v) &\textrm{se }\zeta< v\le1\end{dcases}\label{scarsinipiecewiselinear}\end{equation}
\noindent where $0\le\beta\le\gamma\le\delta\le\epsilon\le\zeta\le1$. As pointed out by the same authors, the function (\ref{scarsinipiecewiselinear}) contains some interesting special cases. While it is a piecewise-linear continuous function if the parameters are all different from each other, it becomes discontinuous, for example, in the particular case where $\beta=\gamma=\delta=\epsilon$ and $\zeta=1$. Figures \ref{fig:scarsinicopula} and \ref{fig:scarsinidisccopula} report the results of some Monte Carlo simulations of copulas constructed using the function (\ref{scarsinipiecewiselinear}), both in the case where it is continuous and in the case where there is a discontinuity.\\
\noindent As a further example, consider the function 
\begin{equation}
g(v)=\begin{dcases}
2v\qquad&\textrm{if }0\le v\le\frac{1}{2}\\
2v-1&\textrm{if }\frac{1}{2}< v\le1.
\end{dcases}\label{gausstypefn}
\end{equation}
It is easily verified that the function (\ref{gausstypefn}) preserves the uniform distribution. An exemplification of a copula constructed by the transformation (\ref{gausstypefn}) is given in figure \ref{fig:modsim}.\\

\begin{figure}[h!]
\centering
\includegraphics[scale=0.75]{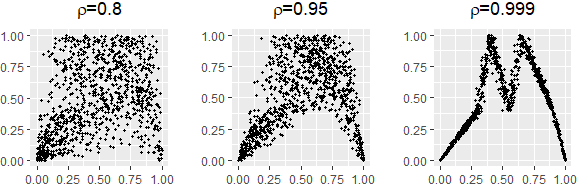}
\caption{Monte Carlo simulations of the pair $(U,g(V))$, where the variable $(U,V)$ has standard uniform marginals and a Gaussian copula with correlation $\rho$ and $g$ is the function (\ref{scarsinipiecewiselinear}) with parameters $\beta=0.3,\;\gamma=0.4,\;\delta=0.55,\;\epsilon=0.65,\;\zeta=0.9$. The simulation was carried out for different values of $\rho$ and for each of the graphs the number of simulated points is 1000.}
\label{fig:scarsinicopula}
\end{figure}

\begin{figure}[h!]
\centering
\includegraphics[scale=0.75]{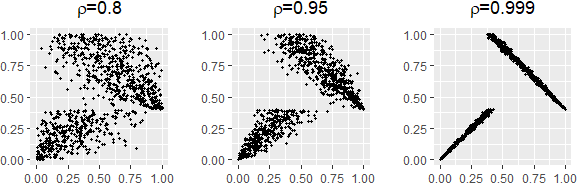}
\caption{Monte Carlo simulations of the pair $(U,g(V))$, where the variable $(U,V)$ has standard uniform marginals and a Gaussian copula with correlation $\rho$ and $g$ is the function (\ref{scarsinipiecewiselinear}) with parameters $\beta=\gamma=\delta=\epsilon=0.4$ and $\zeta=1$. The simulation was carried out for different values of $\rho$ and for each of the graphs the number of simulated points is 1000.}
\label{fig:scarsinidisccopula}
\end{figure}

\begin{figure}[h!]
\centering
\includegraphics[scale=0.75]{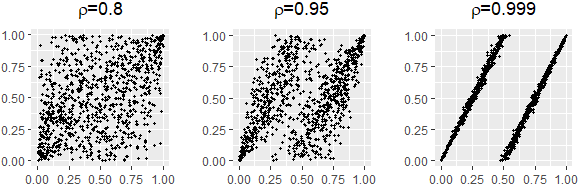}
\caption{Monte Carlo simulations of the pair $(U,g(V))$, where the variable $(U,V)$ has standard uniform marginals and a Gaussian copula with correlation $\rho$ and $g$ is the function (\ref{gausstypefn}). The simulation was carried out for different values of $\rho$ and for each of the graphs the number of simulated points is 1000.}
\label{fig:modsim}
\end{figure}

\section{Non-monotonic dependence copulas}
\noindent As in the cases discussed in the previous section, most of the measure-preserving functions proposed in the literature are piecewise linear. This is primarily due to the fact that constructing measure-preserving functions is not trivial, and the task is significantly simplified by opting for a simpler functional form. The challenge of constructing more general functions which preserve the uniform distribution is relevant from both the theoretical and practical points of view. The development of measure-preserving functions in a sufficiently general form permits to obtain flexible copula models which are able to effectively describe the complicated dependence structures emerging from real-world phenomena. The purpose of constructing general measure-preserving functions is served by the main theorem of our paper, which is stated below (see Marchione \cite{marchione} for details).
 \begin{thm}\label{mainthm}Let $\normalfont{\text{C}}$ be a copula and $f:[0,1]\to[0,1]$ a function such that
\begin{enumerate}
\item $f$ is monotonically increasing with $f(1)=1$
\item $f(v)\ge v,\;\;\forall v\in[0,1]$
\item $\frac{f(v_2)-f(v_1)}{v_2-v_1},\le1\;\;\forall v_1,v_2\in[0,1]$.
\end{enumerate}
Then the function \begin{equation}\label{Q}\normalfont{\text{Q}}(u,v)=\normalfont{\text{C}}(u,f(v))-\normalfont{\text{C}}(u,f(v)-v)\end{equation}is a copula.\end{thm}
\begin{proof}It is trivial to show that $\normalfont{\text{Q}}(u,0)=\normalfont{\text{Q}}(0,v)=0$ and $\normalfont{\text{Q}}(1,v)=v$. Moreover, since $f(1)=1$, we have that $\normalfont{\text{Q}}(u,1)=u$. Thus, in order to prove the theorem, it is sufficient to show that $$\normalfont{\text{V}}_{\normalfont{\text{Q}}}([u_1,u_2]\times[v_1,v_2])\ge0$$
for $u_1\le u_2$ and $v_1\le v_2$. We start by observing that
\begin{equation}\label{thm1_1}\normalfont{\text{V}}_{\normalfont{\text{Q}}}([u_1,u_2]\times[v_1,v_2])=\normalfont{\text{V}}_{\normalfont{\text{C}}}([u_1,u_2]\times[f(v_2)-v_2,f(v_2)])-\normalfont{\text{V}}_{\normalfont{\text{C}}}([u_1,u_2]\times[f(v_1)-v_1,f(v_1)]).\end{equation}
Since $f$ is monotonically increasing and $f(v_2)-v_2\le f(v_1)-v_1$ by hypothesis, we have that $$[f(v_1)-v_1,f(v_1)]\subseteq[f(v_2)-v_2,f(v_2)].$$
This permits us to reformulate equation (\ref{thm1_1}) in the form
$$\normalfont{\text{V}}_{\normalfont{\text{Q}}}([u_1,u_2]\times[v_1,v_2])=\normalfont{\text{V}}_{\normalfont{\text{C}}}([u_1,u_2]\times[f(v_2)-v_2,f(v_1)-v_1])+\normalfont{\text{V}}_{\normalfont{\text{C}}}([u_1,u_2]\times[f(v_1),f(v_2)])\ge0$$
which completes the proof.
\end{proof}
Theorem \ref{mainthm} outlines a method for constructing a copula $\normalfont{\text{Q}}$ starting from a known copula $\normalfont{\text{C}}$ and a function $f$ which satisfies suitable properties. The interesting feature of the proposed method is represented by its capability to generate copulas which represent non-monotonic dependencies. Rather than merely offering examples of measure-preserving functions as previous authors have done, our work introduces a comprehensive method for their systematic construction. In order to clarify this point, consider a bivariate standard uniform random variable $(U,V)$ with copula $\normalfont{\text{C}}$. For a given function $f$ satisfying the conditions of theorem \ref{mainthm}, define the copula $$\normalfont{\text{Q}}(u,v)=\normalfont{\text{C}}(u,f(v))-\normalfont{\text{C}}(u,f(v)-v).$$
Denote by $g^+(x)$ and $g^-(x)$ the pseudo-inverses respectively of $f(v)$ and $f(v)-v$, that is
$$g^+(x)=\inf\{v\in[0,1]:\;f(v)\ge x\}$$
and $$g^-(x)=\sup\{v\in[0,1]:\;f(v)-v\ge x\}.$$
By defining the function \begin{equation}\label{g}
g(x)=\begin{cases}g^+(x)\qquad & x\ge f(0)\\g^-(x)\qquad &x< f(0)\end{cases}
\end{equation}
we can write
\begin{align}\normalfont{\text{Q}}(u,v)=&\normalfont{\text{C}}(u,f(v))-\normalfont{\text{C}}(u,f(v)-v)=\mathbb{P}\left(U\le u,\;f(v)-v<V\le f(v)\right)\nonumber\\
=&\mathbb{P}\left(U\le u,\;g(V)\le v\right).\label{algorithm}\end{align}
\noindent Theorem \ref{mainthm} ensures that $\normalfont{\text{Q}}$ is a copula and this implies, in view of (\ref{algorithm}), that $g(V)\sim\mathcal{U}_{[0,1]}$. In other words, the function $g$ is by construction a measure-preserving function. It can be easily verified that in the trivial cases where $f(0)=1$ or $f(0)=0$, $g$ is a linear function. Excluding these two trivial cases, it can be observed that the behavior of $g$ resembles that of a parabola. It takes a zero value at the point $f(0)$, to the right and left of which it exhibits two upward branches that may display jump discontinuities if $f$ is constant on some intervals (see, for example, figure \ref{fig:fgplot}).
\begin{figure}[h]
\centering
\includegraphics[scale=0.85]{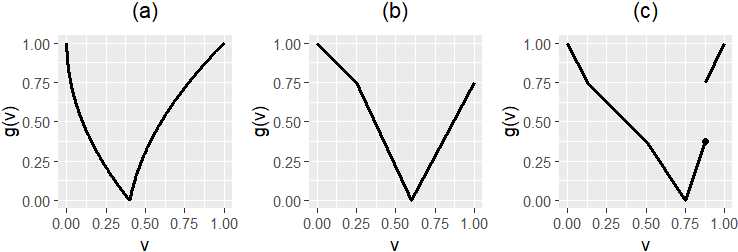}
\caption{Some examples of functions that preserve the uniform distribution obtained through theorem \ref{mainthm} and formula (\ref{g}) for different choices of $f$:\\
(a) $f(v)=\frac{1}{5}(2v^2+v+2)$\\
(b) $f(v)=\min\left(1,\;\frac{8}{15}v+\frac{3}{5}\right)$\\
(c) $f(v)=\min\left(\max\left(\frac{7}{8},\;\frac{1+v}{2}\right),\;\frac{v}{3}+\frac{3}{4}\right)$.}
\label{fig:fgplot}
\end{figure}
\noindent Therefore, $g$ is a non-monotonic function. This implies that, if the copula $\normalfont{\text{C}}$ describes the distribution of two random variables which are stochastically increasing (or decreasing) one with respect to the other, the copula $\normalfont{\text{Q}}$ describes a non-monotonic dependence. The structure of $g$ seems at first glance not to be sufficiently general to capture intricate dependence structures. However, in practice, more complicated dependencies can be described by iteratively applying theorem \ref{mainthm} or, equivalently, by composing measure-preserving functions. Consider, as an example, the functions
$$f_1(v)=1,$$
\vspace{-2mm}

$$ f_2(v)=\max\left(v,\;\frac{\zeta+\gamma-\beta-\epsilon}{\zeta-\beta}v+(\epsilon-\gamma)\right)$$
and
$$f_3(v)=\begin{dcases}
\dfrac{(\epsilon-\delta)}{(\epsilon-\gamma)}\;v +\delta&\textrm{if }0\le v\le\epsilon-\gamma\\
\dfrac{(\zeta-\epsilon)}{(\zeta+\gamma-\epsilon-\beta)}\;v -\frac{(\epsilon-\gamma)(\zeta-\epsilon)}{(\zeta+\gamma-\epsilon-\beta)}+\epsilon\qquad&\textrm{if }\epsilon-\gamma<v\le\zeta-\beta\\
\dfrac{(1-\zeta)}{(1-\zeta+\beta)}\;v +\frac{\beta}{(1-\zeta+\beta)}&\textrm{if }\zeta-\beta<v\le1.
\end{dcases}$$
\noindent which satisfy, for $0\le\beta\le\gamma\le\delta\le\epsilon\le\zeta\le1$, the assumptions of theorem \ref{mainthm}. By defining the corresponding measure-preserving functions as in formula (\ref{g}), that is
\begin{equation*}
g_j(x)=\begin{cases}g_j^+(x)\qquad & x> f_j(0)\\g_j^-(x)\qquad &x\le f_j(0)\end{cases}\qquad j=1,2,3
\end{equation*}
where $g_j^+(x)$ and $g_j^-(x)$ represent the pseudo-inverses respectively of $f_j(v)$ and $f_j(v)-v$, it can be verified that the transformation $$g=g_1\circ g_2\circ g_3$$
\noindent coincides with the function (\ref{scarsinipiecewiselinear}) examined by Scarsini and Venetoulias \cite{scarsini}. Similarly, by setting $$f_1(v)=\max\left(1,\;v\right),\qquad f_2(v)=\frac{1+v}{2}$$
and defining the corresponding measure-preserving transformations $g_j,\;j=1,2$ as in (\ref{g}), it follows that the function $$g=g_1\circ g_2$$ coincides with the measure-preserving function (\ref{gausstypefn}). We can therefore conclude that theorem \ref{mainthm} provides a flexible tool for the construction of copulas that describe dependencies of arbitrary shape. As we will discuss in the remainder of the paper, it is possible to construct such copulas in parametric form simply by constructing parametric functions which preserve the uniform distribution.

\section{Parameter estimation and model selection}\label{methodology}
\noindent Consider a monotonic-type copula $\normalfont{\text{C}}_{\boldsymbol{\theta}}:[0,1]^2\to[0,1]$ depending on a vector of parameters $\boldsymbol{\theta}$ and a function $f_{\boldsymbol{\alpha}}:[0,1]\to[0,1]$ which depends on a vector of parameters $\boldsymbol{\alpha}$. If $f_{\boldsymbol{\alpha}}$ satisfies the assumptions of theorem \ref{mainthm}, we can construct a copula $\normalfont{\text{Q}}_{\boldsymbol{\alpha},\boldsymbol{\theta}}$ as
\begin{equation}\label{Qtoest}\normalfont{\text{Q}}_{\boldsymbol{\alpha},\boldsymbol{\theta}}(u,v)=\normalfont{\text{C}}_{\boldsymbol{\theta}}(u,f_{\boldsymbol{\alpha}}(v))-\normalfont{\text{C}}_{\boldsymbol{\theta}}(u,f_{\boldsymbol{\alpha}}(v)-v),\qquad u,v\in[0,1].\end{equation}
\noindent Assuming that the copula $\normalfont{\text{C}}_{\boldsymbol{\theta}}$ admits the density $\normalfont{\text{c}}_{\boldsymbol{\theta}}$ and that $f_{\boldsymbol{\alpha}}$ is differentiable, we can express the density of $\normalfont{\text{Q}}_{\boldsymbol{\alpha},\boldsymbol{\theta}}$ in the form
\begin{equation}\label{density}\normalfont{\text{q}}_{\boldsymbol{\alpha},\boldsymbol{\theta}}(u,v)=\left[\normalfont{\text{c}}_{\boldsymbol{\theta}}(u,f_{\boldsymbol{\alpha}}(v))-\normalfont{\text{c}}_{\boldsymbol{\theta}}(u,f_{\boldsymbol{\alpha}}(v)-v)\right]\frac{df_{\boldsymbol{\alpha}}(v)}{dv}+\normalfont{\text{c}}_{\boldsymbol{\theta}}(u,f_{\boldsymbol{\alpha}}(v)-v).\end{equation}
\noindent The probability density function (\ref{density}) permits to perform a maximum likelihood estimation of the parameter vectors  $\boldsymbol{\alpha}$ and $\boldsymbol{\theta}$. Consider the sample $(u_1,v_1), (u_2,v_2),..., (u_n,v_n)$ coming from a bivariate distribution of the form $\normalfont{\text{Q}}_{\boldsymbol{\alpha},\boldsymbol{\theta}}$ where the parameters $\boldsymbol{\alpha},\boldsymbol{\theta}$ are unknown. By maximizing the log-likelihood function $$\ell(\boldsymbol{\alpha},\boldsymbol{\theta};\;\mathbf{u},\mathbf{v})=\sum_{k=1}^n\log\normalfont{\text{q}}_{\boldsymbol{\alpha},\boldsymbol{\theta}}(u_k,v_k)$$
with respect to the parameter vectors $\boldsymbol{\alpha},\boldsymbol{\theta}$, the maximum likelihood estimate of the parameters is obtained
$$(\widehat{\boldsymbol{\alpha}},\widehat{\boldsymbol{\theta}})=\arg\max_{\boldsymbol{\alpha},\boldsymbol{\theta}}\;\ell(\boldsymbol{\alpha},\boldsymbol{\theta};\;\mathbf{u},\mathbf{v}).$$
In order to perform model selction, the \textit{Akaike Information Criterion} (AIC) and the \textit{Bayesian Information Criterion} (BIC) can be employed. First introduced by Akaike \cite{aic} and Schwarz \cite{bic}, these indicators are respectively defined as
\begin{equation*}
\text{AIC} = 2k - 2\widehat{\ell}
\end{equation*}
and
\begin{equation*}
\text{BIC} = \log(n)\,k - 2\widehat{\ell}
\end{equation*}
where \(\widehat{\ell} = \ell(\widehat{\boldsymbol{\alpha}}, \widehat{\boldsymbol{\theta}}; \mathbf{u}, \mathbf{v})\) denotes the maximized value of the likelihood function, \(k\) represents the number of model parameters and $n$ is the sample size. In practice, parameter estimation is performed for different copulas \(\normalfont{\text{Q}}_{\boldsymbol{\alpha}, \boldsymbol{\theta}}\) and then the AIC or BIC scores are calculated for each estimated copula. The copula model to be selected is the one with the lowest score. This approach favors models with higher likelihood, indicating better data fit, while penalizing those with more parameters, thus leaning towards a more parsimonious model choice. It can be observed that the BIC criterion imposes a higher penalty on models with a large number of parameters compared to the AIC criterion.

\section{Data and results}\label{sec:datierisultati}
\noindent We analyze the historical time series of daily closing prices and trading volumes for the QQQ fund shares, an investment fund managed by Invesco Ltd. which is designed to track the Nasdaq 100 index. Specifically, we consider \( n+1 \) observations, with \( n=624 \), covering the period from February 1st 2020 to June 24th 2022. Let \( p_k, \; k=0,...,n, \) represent the \( k \)-th observation of the daily closing price. We calculate the daily returns $\mathbf{r}=(r_1,...,r_n)$ as
\begin{equation*}
r_k = \frac{p_k - p_{k-1}}{p_{k-1}}, \qquad k=1,...,n.
\end{equation*}
and we denote by \( \mathbf{w}=(w_1,...,w_n) \) the corresponding trading volumes. In order to investigate the volume-return dependence structure we first apply the probability integral transform to both the variables. This permits to transform the  distributions of $\mathbf{r}$ and $\mathbf{w}$ into uniform distributions, therefore allowing to apply the methodology described in section \ref{methodology}. We use an empirical version of the probability integral transform by computing the pseudo-observations
$$u_k=\frac{1}{n+1}\sum_{j=1}^n\mathbb{1}(r_j\le r_k)=\frac{n}{n+1}\widehat{F}_r(r_k),\qquad k=1,...,n$$
$$v_k=\frac{1}{n+1}\sum_{j=1}^n\mathbb{1}(w_j\le w_k)=\frac{n}{n+1}\widehat{F}_w(w_k),\qquad k=1,...,n$$
\noindent where $\widehat{F}_r(\cdot)$ and $\widehat{F}_w(\cdot)$ are the empirical cumulative distribution functions of the returns and the exchange volumes respectively.\\

\begin{figure}[h]
\centering
\includegraphics[scale=0.9]{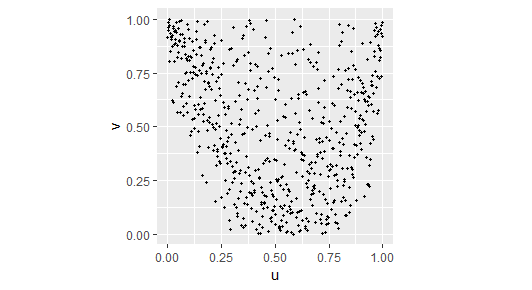}
\caption{Graph of the points corresponding to the pseudo-observations of daily returns (on the $x$-axis) and daily trading volumes (on the $y$-axis) for the QQQ fund. The graph suggests the presence of a non-monotonic type of dependence.}
\label{fig:data}
\end{figure}
\noindent Figure \ref{fig:data} shows the scatterplot of the transformed exchange volumes $\mathbf{v}=(v_1,...,v_n)$ versus the transformed returns $\mathbf{u}=(u_1,...,u_n)$. Unlike the case of monotonic dependencies, where the points tend to be located around the diagonal of the unit square, in this case the points of the scatterplot form a smile which suggests a non-monotonic dependence between returns and exchange volumes. In particular, we observe that exchange volumes tend to be high both when returns are high and low. This conclusion is consistent with the majority of the literature concerning the volume-return relationship. For a detailed review of the literature on the topic we refer to Karpoff \cite{karpoff}. Clearly, the classical monotonic-type copulas are not suitable for modeling such a dependence structure.\\
\noindent In order to describe the non-monotonic dependence emerging from the considered data, we use the results of the previous sections to construct and fit appropriate parametric copulas. We perform parameter estimation for copulas of the form (\ref{Qtoest}) for different choices of $\normalfont{\text{C}}_{\boldsymbol{\theta}}$ and $f_{\boldsymbol{\alpha}}$. In particular, we compare the Frank, Clayton, Gumbel, Gaussian and t copulas. For each of the mentioned copulas we propose three possible choices of the function $f_{\boldsymbol{\alpha}}$, namely
\begin{align}
f_1(&v;c)=\frac{v+c}{1+c},\qquad &&c\ge0\nonumber\\
f_2(&v;a, c)=av^2+(1-a-c)v+c,\qquad && c\ge0,\;0\le a+c\le1,\;0\le1+a-c\le1\nonumber\\
f_3(&v;a,c)=\left(\frac{v+c}{1+c}\right)^a,\qquad && a\ge0,\;c\ge0,\;a\le1+c,\;\left(\frac{c}{1+c}\right)^{a-1}\le\frac{1+c}{a}\nonumber\end{align}
\noindent where the constraints are imposed in such a way that the functions satisfy the assumptions of theorem \ref{mainthm}. Therefore, we fit a total of fifteen different models with a number of parameters ranging from two to four.  We emphasize that the constraints to be imposed for $f_1$, $f_2$ and $f_3$ have different nature. While the constraint for $f_1$ can be easily imposed by a simple reparametrization, the parameters of $f_2$ and $f_3$ are subject to linear and non-linear constraints respectively. This implies that suitable optimization algorithms must be chosen for each case.  The results of the parameter estimation procedure are presented in Table \ref{MLparams_table}, while the corresponding AIC and BIC indices are detailed in Tables \ref{AIC_table} and \ref{BIC_table}, respectively.

\begin{table}[h!]
\centering
\begin{tabular}{| c | c | c | c |}
\cline{2-4}
\multicolumn{1}{c|}{} & $f_1$ & $f_2$ & $f_3$ \\
\hline
Clayton & \makecell{c=1.179097\\$\theta$=1.7217665} &  \makecell{a=0.003634653\\c=0.5430549\\$\theta$=1.7205763}  & \makecell{a=1.0297248\\c=1.2363275\\$\theta$=1.7206042}  \\
\hline
Gumbel & \makecell{c=2.241228\\$\theta$=2.2659783} & \makecell{a=-0.226820627\\c=0.5937735\\$\theta$=2.2885958}   & \makecell{a=0.2779444\\c=0.1603452\\$\theta$=2.2945399}  \\
 \hline  
Frank & \makecell{c=1.460967\\$\theta$=9.0344146} & \makecell{a=-0.159780111\\c=0.5415870\\$\theta$=9.0850997}    & \makecell{a=0.4529352\\c=0.3375023\\$\theta$=9.0684555}  \\
\hline 
Gaussian & \makecell{c=1.815000\\$\theta$=0.7826850} & \makecell{a=-0.186254604\\c=0.5636995\\$\theta$=0.7849017}  & \makecell{a=0.3397894\\c=0.1813854\\$\theta$=0.7848028}  \\
\hline
t & \makecell{c=1.787246\\$\theta$=0.7751602\\$\nu$=5.884596} & \makecell{a=0.204120586\\c=0.4411932\\$\theta$=-0.7756732\\$\nu$=5.257317}  & \makecell{a=0.3376143\\c=0.1838309\\$\theta$=0.7767609\\$\nu$=5.517027}  \\
\hline 
\end{tabular}
\caption{\label{MLparams_table}Results of the maximum likelihood estimation for different combinations of $\normalfont{\text{C}}_{\boldsymbol{\theta}}$ and $f_{\boldsymbol{\alpha}}$. In the case of the t copula, the additional parameter $\nu$ represents the degrees of freedom parameter.}
\vspace{8mm}
\centering
\begin{tabular}{| c | c | c | c |}
\cline{2-4}
\multicolumn{1}{c|}{} & $f_1$ & $f_2$ & $f_3$ \\
\hline
Clayton & -171.1221 & -169.1243  & -169.1243 \\
\hline
Gumbel & -218.6425 & -236.0971  & -233.3283 \\
 \hline  
Frank & -234.7418 & \textbf{-243.0672}  & -241.2312 \\
\hline 
Gaussian & -214.0477 & -220.5011  & -219.6703 \\
\hline
t & -229.4568 & -238.4797  & -236.8192 \\
\hline 
\end{tabular}
\caption{\label{AIC_table}AIC scores for the examined copulas.}
\vspace{8mm}
\centering
\begin{tabular}{| c | c | c | c |}
\cline{2-4}
\multicolumn{1}{c|}{} & $f_1$ & $f_2$ & $f_3$ \\
\hline
Clayton & -162.2498 & -155.8158  & -155.8159 \\
\hline
Gumbel & -209.7702 & -222.7886  & -220.0199 \\
 \hline  
Frank & -225.8695 & \textbf{-229.7588}   & -227.9227 \\
\hline 
Gaussian & -205.1754 & -207.1927  & -206.3618 \\
\hline
t & -216.1483 & -220.7351  & -219.0746 \\
\hline 
\end{tabular}
\caption{\label{BIC_table}BIC scores for the examined copulas.}
\end{table}




\noindent These tables clearly indicate that the quadratic function \( f_2 \), when paired with the Frank copula, emerges as the optimal choice in terms of both AIC and BIC for modeling the dependency structure between returns and volumes in the examined dataset. It's noteworthy that the linear function \( f_1 \) is essentially a special case of \( f_2 \), derived by setting \( a=0 \). This implies that the model selection criterion based on AIC and BIC justifies the incorporation of the additional parameter $a$ to enhance the model's fit to the data. Indeed, the increase in AIC and BIC values due to the penalization for the extra parameter is more than offset by the improvement in likelihood. In Figure \ref{fig:QQQcurve}, we illustrate the measure-preserving function derived from the quadratic function \( f_2 \) which, according to the parameter estimation results, most accurately approximates the dependency between volume and return for the dataset under study.\\

\begin{figure}[h]
\centering
\includegraphics[scale=0.85]{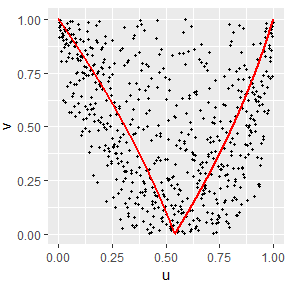}
\caption{Graph of the points corresponding to the pseudo-observations of daily returns and daily trading volumes. The red curve coincides with the measure-preserving function implicit in the estimated model combining the Frank copula and the quadratic function \( f_2 \).}
\label{fig:QQQcurve}
\end{figure}

For the sake of completeness, we now estimate the parameters of monotonic-type copulas to compare and verify that copulas representing non-monotonic dependencies are more suitable for describing the considered data. To this end, we use the R package \texttt{VineCopula}, particularly the \texttt{BiCopSelect} function. This function estimates parameters for various elliptical and Archimedean copulas and their generalizations, including possible rotations of the distributions. In total, thirtyseven models are compared. For a detailed list, we refer to the documentation by Schepsmeier et al. \cite{vinecopula}. The \texttt{BiCopSelect} function outputs the best-fitting model according to user-defined criteria. The best copula among those compared, in terms of both AIC and BIC, is found to be the $90^\circ$ rotated Tawn type 1 copula. The Tawn copula (see Tawn \cite{tawn} and Bernard and Czado \cite{bernard}) is a three-parameter copula, with parameters $\theta, \psi_1, \psi_2$ subject to the constraints $\psi_1, \psi_2 \in [0,1]$ and $\theta \in [1, +\infty)$, defined as
\begin{equation}
\normalfont\textrm{C}_{\theta,\psi_1,\psi_2}(u,v)=(uv)^{A_{\theta,\psi_1,\psi_2}\left(\frac{\log(u)}{\log(uv)}\right)}\label{tawncopula}
\end{equation}
where 
\begin{equation*}
A_{\theta,\psi_1,\psi_2}(t)=(\psi_2-\psi_1)t+(1-\psi_2)+\left([\psi_2(1-t)]^{\theta}+(\psi_1 t)^{\theta}\right)^{\frac{1}{\theta}}.
\end{equation*}
For \( \psi_1 = \psi_2 = 1 \), the Gumbel copula is obtained. The Tawn type 1 copula refers to the bi-parametric copula obtained by setting \( \psi_2 = 1 \) in formula (\ref{tawncopula}). Additionally, after a $90^\circ$ rotation, the resulting copula can be expressed as 
\begin{equation}
\normalfont\textrm{C}_{\theta,\psi_1}^{Rot}(u,v)=v-\normalfont\textrm{C}_{\theta,\psi_1,1}(1-u,v).\label{tawncopularot}
\end{equation}
The parameters estimated by the \texttt{BiCopSelect} function for the copula (\ref{tawncopularot}) are $\theta = 2.5905362$ and $\psi_1 = 0.2815409$, while the AIC and BIC indices are $\textrm{AIC} = -155.1409$ and $\textrm{BIC} = -146.2686$. We reiterate that the selected copula is the best in terms of AIC and BIC among thirty-seven compared monotonic copulas. This copula has higher AIC and BIC values than all the non-monotonic copulas we previously estimated (see tables \ref{AIC_table} and \ref{BIC_table}). This result supports the conclusion of a non-monotonic dependency between volume and return.

\section{Conclusions}
\noindent The results presented in this paper establish a comprehensive method for creating copulas which are capable of representing dependencies of arbitrary form. As highlighted by Scarsini and Venetoulias \cite{scarsini}, the definition of non-monotonic dependence copulas is closely related to the construction of functions that possess the property of preserving the uniform distribution. The here presented results go beyond the existing literature by enabling the construction of parametric functions which satisfy this property. This advancement endows the resultant copula models with an improved flexibility which permits to accurately describe real-world phenomena with intricate dependence structures.\\

\noindent The practicality and adaptability of the proposed methodology are exemplified through its application to the relationship between returns and trading volumes in financial markets, a phenomenon for which the non-monotonic dependence is well-established in the existing literature. Generalizing the work by Neto \cite{neto}, we constructed copulas involving parameters which are able to accurately model the shape of the examined dependence structure. The here adopted model selection criteria validate the inclusion of such parameters, affirming the advantage of utilizing more versatile parametric models for a deeper and more precise description of the data.


\bibliographystyle{plain}
\nocite{*}
\bibliography{bibliography}

\begin{thebibliography}{10}

\bibitem{aic}
H.~Akaike.
\newblock Information theory and an extension of the maximum likelihood
  principle.
\newblock In {\em Petrov, B.N. and Csaki, F., Eds., International Symposium on
  Information Theory}, pages 267--281, 1973.

\bibitem{bernard}
Carole Bernard and Claudia Czado.
\newblock Conditional quantiles and tail dependence.
\newblock {\em Journal of Multivariate Analysis}, 138:104--126, 2015.

\bibitem{clark}
Peter~K. Clark.
\newblock A subordinated stochastic process model with finite variance for
  speculative prices.
\newblock {\em Econometrica}, 41(1):135--155, 1973.

\bibitem{crouch}
R.~L. Crouch.
\newblock A nonlinear test of the random-walk hypothesis.
\newblock {\em The American Economic Review}, 60(1):199--202, 1970.

\bibitem{embrechts}
Paul Embrechts.
\newblock Correlation: pitfalls and alternatives.
\newblock {\em Risk Magazine}, pages 69--71, 1999.

\bibitem{frees}
Edward~W. Frees and Emiliano~A. Valdez.
\newblock Understanding relationships using copulas.
\newblock {\em North American Actuarial Journal}, 2(1):1--25, 1998.

\bibitem{frechet}
M.~Fréchet.
\newblock Sur les tableaux dont les marges et des bornes sont données.
\newblock {\em Revue de l'Institut International de Statistique / Review of the
  International Statistical Institute}, 28(1/2):10--32, 1960.

\bibitem{genest}
Christian Genest, Michel Gendron, and Michaël Bourdeau-Brien.
\newblock The advent of copulas in finance.
\newblock {\em The European Journal of Finance}, 15(7-8):609--618, 2009.

\bibitem{gunduz}
Lokman Gündüz and Abdulnasser Hatemi-J.
\newblock Stock price and volume relation in emerging markets.
\newblock {\em Emerging Markets Finance and Trade}, 41(1):29--44, 2005.

\bibitem{joe}
Harry Joe.
\newblock {\em Multivariate models and dependence concepts}.
\newblock Chapman and Hall, London, 1997.

\bibitem{karpoff}
Jonathan~M. Karpoff.
\newblock The relation between price changes and trading volume: A survey.
\newblock {\em The Journal of Financial and Quantitative Analysis},
  22(1):109--126, 1987.

\bibitem{lee}
Bong-Soo Lee and Oliver~M. Rui.
\newblock The dynamic relationship between stock returns and trading volume:
  Domestic and cross-country evidence.
\newblock {\em Journal of Banking \& Finance}, 26(1):51--78, 2002.

\bibitem{marchione}
Manfred~M. Marchione.
\newblock {\em Copule per la modellizzazione di dipendenze non monotone:
  applicazioni in finanza e demografia}.
\newblock PhD thesis, Università degli Studi di Roma "La Sapienza", 2024.

\bibitem{nelsenold}
Roger~B. Nelsen.
\newblock {\em An introduction to copulas}.
\newblock Springer, New York, 2 edition, 1999.

\bibitem{neto}
David Neto.
\newblock Dépendance non-monotone: Une application à la relation
  rendement-volume.
\newblock {\em Annals of Economics and Statistics}, (82):187--216, 2006.

\bibitem{osborne}
M.~F.~M. Osborne.
\newblock Brownian motion in the stock market.
\newblock {\em Operations Research}, 7(2):145--173, 1959.

\bibitem{scarsini}
Marco Scarsini and Achilles Venetoulias.
\newblock Bivariate distributions with nonmonotone dependence structure.
\newblock {\em Journal of the American Statistical Association},
  88(421):338--344, 1993.

\bibitem{vinecopula}
Ulf Schepsmeier, Jakob Stoeber, Eike~Christian Brechmann, Benedikt Graeler,
  Thomas Nagler, Tobias Erhardt, Carlos Almeida, Aleksey Min, Claudia Czado,
  Mathias Hofmann, Matthias Killiches, Harry Joe, and Thibault Vatter.
\newblock Package ‘vinecopula’.
\newblock {\em R package version}, 2(5), 2015.

\bibitem{bic}
Gideon Schwarz.
\newblock Estimating the dimension of a model.
\newblock {\em The Annals of Statistics}, 6(2):461--464, 1978.

\bibitem{sklar}
Abe Sklar.
\newblock {Fonctions de Répartition à n Dimensions et Leurs Marges}.
\newblock {\em Publications de l’Institut Statistique de l’Université de
  Paris}, 8:229--231, 1959.

\bibitem{tauchen}
George~E. Tauchen and Mark Pitts.
\newblock The price variability-volume relationship on speculative markets.
\newblock {\em Econometrica}, 51(2):485--505, 1983.

\bibitem{tawn}
Jonathan~A. Tawn.
\newblock Bivariate extreme value theory: Models and estimation.
\newblock {\em Biometrika}, 75(3):397--415, 1988.

\end{thebibliography}
\end{document}